\documentclass[preprint,12pt]{elsarticle}

\usepackage{hyperref}
\usepackage{graphicx}
\usepackage{float}
\usepackage{enumerate}
\usepackage{mathrsfs}
\usepackage{color}
\usepackage{mathtools}
\usepackage{breqn}
\usepackage{subfigure}
\usepackage{indentfirst}
\usepackage{arydshln}
\usepackage{cases}
\usepackage{amsthm}
\usepackage{amsmath,amssymb,amsfonts}
\newtheorem{example}{Example}

\newtheorem{question}{Question}

\newtheorem{theorem}{Theorem}
\newtheorem{lemma}{Lemma}
\newtheorem{definition}{Definition}
\newtheorem{remark}{Remark}

\newtheorem{algorithm}{Algorithm}
\newcommand{\defword}{\textbf}

\makeatletter
\def\ps@pprintTitle{%
  \let\@oddhead\@empty
  \let\@evenhead\@empty
  \def\@oddfoot{\hfil}
  \let\@evenfoot\@oddfoot
}
\makeatother

\begin{document}

\begin{frontmatter}



\title{Determining the Equivalence of Small Zero-one Reaction Networks}


\author{Yue Jiao\fnref{aff1,aff2}}
\author{Xiaoxian Tang\fnref{aff1}}
\fntext[aff1]{School of Mathematical Sciences, Beihang University, Beijing, 100191, China}
\fntext[aff2]{Max Planck Institute of Molecular Cell Biology and Genetics, Dresden, 01307, Germany}

\begin{abstract}
Zero-one reaction networks are pivotal to cellular signaling, and establishing the equivalence of such networks represents a foundational computational challenge in the realm of chemical reaction network research. Herein, we propose a high-efficiency approach for identifying the equivalence of zero-one networks. Its efficiency stems from a set of criteria tailored to judge the equivalence of steady-state ideals derived from zero-one networks, which effectively reduces the computational cost associated with Gr\"obner basis calculations. Experimental results demonstrate that our proposed method can successfully categorize more than three million networks by their equivalence within a feasible timeframe. Also, our computational results for two important classes of quadratic zero-one networks ($3$-dimensional with $3$ species, $6$ reactions; $4$-dimensional with $4$ species, $5$ reactions) show that they have no positive steady states for a generic choice of rate constants, implying these small networks generically exhibit neither multistability nor periodic orbits.    
\end{abstract}




\end{frontmatter}

\maketitle

\section{Introduction}
The dynamical systems generated by biochemical reaction networks give rise to the following fundamental computational question.
\textcolor{black}{
\begin{question}\label{question}
Which reaction networks have the same steady-state ideal?
\end{question}
}
The  congruence of steady-state ideals is one of the core equivalence problems in the field of  chemical reaction networks.  Numerous key dynamical characteristics of chemical reaction networks—including multistability, Hopf bifurcations, and absolute concentration robustness (ACR)—are correlated with switch-like behaviors, decision-making processes, oscillations, and other key events in cellular signaling
\cite{FA2020, SE2016, XK2023, CE2019, MB2023, LE2025}.
Typically, examining the steady states constitutes the primary step in investigating the dynamical behaviors of a chemical reaction system  arising under mass-action kinetics. 
Notice that  the  steady states can be considered as  positive real points located in the algebraic variety generated by the steady-state ideal.
However, it is possible for two distinct networks to have exactly the same steady-state ideal. For instance, consider the following two reaction networks \eqref{eq:network1} and \eqref{eq:network2}.
\begin{align}
 \label{eq:network1} &X_1+X_2\xrightarrow{\kappa_1}0,  
 &&X_2 \xrightarrow{\kappa_2}X_1+X_2, 
 &&0\xrightarrow{\kappa_3}X_1+X_2,
 &&0\xrightarrow{\kappa_4}X_2\\
\label{eq:network2}
&X_1+X_2\xrightarrow{\kappa_1}X_2,  
&&X_2 \xrightarrow{\kappa_2}X_1,  
&&0\xrightarrow{\kappa_3}X_1,
&&0\xrightarrow{\kappa_4}X_2
\end{align}
Their steady-state ideals are respectively generated by the following two sets of polynomials \eqref{eq:poly1} and \eqref{eq:poly2} (we will explain how to write down these polynomials for the given networks later in \eqref{eq:sys}).
\begin{align}
&\{-\kappa_1x_1x_2+\kappa_2x_2+\kappa_3, -\kappa_1x_1x_2+\kappa_3+\kappa_4\}\label{eq:poly1}\\
&\{-\kappa_1x_1x_2+\kappa_2x_2+\kappa_3, -\kappa_2x_2+\kappa_4\}
\label{eq:poly2}
\end{align}
It is straightforward to compute that   the above two steady-state systems have the same  reduced  Gr\"obner basis  \cite{BB1965} in $\mathbb Q(\kappa)[x]$: 
\begin{align}
\{\kappa_2x_2 - \kappa_4, \kappa_1\kappa_4x_1 - \kappa_2\kappa_3 - \kappa_2\kappa_4\}.
\end{align}
In numerous instances, the two networks above are regarded equivalent if our focus is solely on dynamical behaviors governed by steady states, including multistationarity, multistability, Hopf bifurcations, and ACR. 
Notice that relabeling the species or the reactions in a network naturally gives another kind of equivalence, say natural equivalence. 
For example, if we relabel the species $X_1$ and $X_2$ and we relabel the first and the second reactions in network \eqref{eq:network1}, then we get the following network 
\begin{align}\label{eq:network3}
 X_1 \xrightarrow{\kappa_1}X_1+X_2,\;\; X_1+X_2\xrightarrow{\kappa_2}0,\; \;
 0\xrightarrow{\kappa_3}X_1+X_2,\;\;
 0\xrightarrow{\kappa_4}X_1.
\end{align}
Of course, network \eqref{eq:network3} should be considered  equivalent to network \eqref{eq:network1}. So, it is also equivalent to network \eqref{eq:network2}. Hence, in our context, we say that two networks are equivalent if they have the same steady-state ideal after undergoing relabeling, see Definition \ref{def:equivalent} (for further details on the ``dynamical equivalence" of networks, refer to \cite{MB2024}). Determining equivalence among a large class of reaction networks is computationally challenging, as the standard approach involves computing their respective reduced  Gr\"obner bases.
In a recent related work \cite{ME2020},  the authors point out that there are too few works investigating the equivalence of steady-state ideals, and we also see that deriving explicit conditions for the equivalence is  difficult in general.

In this work, we propose focusing on a specific class of small networks to make this problem more feasible.   The  idea of   studying small networks is motivated by the fact that many important dynamical features {are inherited from} a large network to a smaller ``subnetwork" such as multistability  \cite{JA2013, BP16}, oscillation \cite{M2018} and local bifurcations \cite{MBJ2023}. A growing body of recent research focuses on finding the smallest networks from extensive classes of networks that admit these dynamical behaviors, and determining network equivalence first is essential to simplify the search process \cite{MB2023,XK2023,KN2024,YX2025}.  
 Although computing the reduced Gr\"obner basis  \cite{WP2022} for a single small network  might not be difficult,   determining equivalence in the aforementioned studies  remains  challenging,  as one typically needs to handle a substantial amount of networks. 
  
  Another core proposal of this work is to focus specifically on zero-one networks for equivalence determination.
 Zero-one  networks are  crucial for intracellular signal transduction primarily such as the phosphorylation  cycle \cite{HSP2008}, cell cycle \cite{NTJ2022}, hybrid histidine kinase \cite{SRS2023}, and so on. They inherently ensure signal robustness, enabling stable responses against minor internal or external fluctuations instead of ambiguous intermediate outputs, which is essential for cells to make accurate functional decisions (e.g., proliferation, apoptosis).  

In this work, our main contribution is   an efficient algorithm (Algorithm \ref{alg:main} with  sub-algorithms Algorithms \ref{alg:inconsistent}--\ref{alg:gb1}) for determining the equivalence of zero-one networks {that admit positive steady states}. The efficiency of the proposed new method comes from several criteria for determining the equivalence of the steady-state ideals arising from zero-one networks, which dramatically  cuts down the expenses on computing Gr\"obner bases since many computations are avoided.  {Note that zero-one networks are enumerable when the numbers of species  and reactions are given. We make use of this fact to 
generate the possible rational vectors ``$a$" stated in  Theorem \ref{thm:coe N}, enabling efficient verification of whether a network admits positive steady states, see Algorithm \ref{alg:inconsistent}. 
By a similar idea, we generate  the possible rational vectors ``$a$" stated in condition (I) of Theorem \ref{lm:gb 1 para} to efficiently check if a steady-state ideal is vacuous, see Algorithm \ref{alg:gb1}. In some special cases, the ``zero-one" assumption can be  used to develop simple criteria for determining the vacuous ideals, see Theorem \ref{lm:gb1 col}.}  {Consequently, our proposed method obviates the necessity of Gr\"obner basis computations, and accurately excludes both networks that do not admit positive steady states and those with vacuous ideals, thus enabling highly efficient, precise pre-screening of target networks.
 }

 This paper reports improved and extended research work relative to the preliminary version published at ISSAC 25, with specific advancements as follows. 
 First, in terms of algorithm design and implementation,  we make the following two main improvements.
 \begin{itemize}
     \item We have refined the first three steps of the main algorithm  \cite[Algorithm 
 \ref{alg:main}]{YX2025}. We enumerate stoichiometric matrices instead of reaction networks, motivated by two key observations: a single stoichiometric matrix typically maps to numerous distinct networks, and preliminary steps including consistency verification are only necessary for stoichiometric matrices. 
 \item We develop Theorem \ref{thm:coe N} and Algorithm \ref{alg:inconsistent} to verify consistency, given that most input networks do not admit positive steady states and we aim to filter them out prior to determining equivalence.  
 \end{itemize}
  We have retested the new algorithm to  determine  the equivalence of  the smallest zero-one networks that admit multistationarity (three-dimensional, three-species and five-reaction \cite{YX2024}). 
  As reported in \cite{YX2025}, this class 
   contains over three millions  networks, and the new algorithm returns  {$38185$} networks after determining the equivalence. 
Experimental results confirm the new algorithm achieves $50\%$ higher efficiency, halving computational time relative to the old version. A comparison between \cite[Table 1]{YX2025} and Table \ref{ta:main} reveals a reduction in total computational time from $1$ hour to less than $0.5$ hours. It is seen from Table \ref{ta:main} that although the total number of tested networks exceeds three millions, the number  of corresponding stoichiometric matrices is merely $43520$. After implementing the first three steps, the number is further reduced to $2548$ as  the matrices that admit no positive steady states are excluded and  the natural equivalence is checked.  At this stage, the  corresponding network count drops to around $100,000$. Notably, the first three steps take only $3$ minutes to complete, while the old version took roughly half an hour to perform these preliminary operations.

  Second, from an application standpoint, we further apply our algorithm to two important families of small zero-one networks: three-dimensional systems (three species with six reactions) and four-dimensional systems (four species with five reactions). Our overarching goal is to identify all minimal zero-one networks that admit multistability \cite{YX2024} and Hopf bifurcations \cite{XK2023} from these two families, respectively. Given the substantial size of each family, our current testing efforts have focused exclusively on quadratic networks. Computational results show that no networks remain in either family after excluding inconsistent networks, those with vacuous steady-state ideals, and those with monomial steady-state ideals (see Tables \ref{tab:363}–\ref{tab:454}). This implies that for generic  choices of rate constants, such small quadratic zero-one networks lack positive steady states and thus cannot display nontrivial dynamical behaviors (e.g., multistability and periodic orbits). Notably, this finding aligns with the established result for two-dimensional quadratic networks \cite[Lemma 3]{MBJ2023}, suggesting that dynamically nontrivial small zero-one networks must be at least cubic.
 

The rest of this paper is organized as follows.
 Section~\ref{sec:background} reviews core concepts of reaction networks and the dynamical systems derived from mass-action kinetics, while formally defining the steady-state ideal and network equivalence within our framework.
Section~\ref{sec:thm} establishes a suite of practical criteria—including Theorems \ref{thm:coe N}–\ref{lm:gb1 col} and Lemma \ref{lm:ideal equal}—for determining network equivalence.
Section~\ref{sec:alg} presents the main algorithm (Algorithm \ref{alg:main}) alongside its auxiliary sub-algorithms (Algorithms \ref{alg:inconsistent}–\ref{alg:gb1}) tailored to verifying the equivalence of zero-one networks.
Section~\ref{sec:implement} demonstrates the main algorithm’s performance via applications to three-dimensional zero-one networks (three species, five reactions), as well as two families of quadratic zero-one networks: three-dimensional (three species, six reactions) and four-dimensional (four species, five reactions). This section further elaborates on implementation details and reports computational runtime results.

\section{Background}\label{sec:background}
In this section, we briefly recall the standard notions and definitions for reaction networks and computational algebraic geometry, see \cite{CFMW}  and \cite{DC1997} for more details. 

A \defword{reaction network} $G$  (or \defword{network} for short) consists of a set of $s$ species $\{X_1, X_2, \dots, X_s\}$ and a set of $m$ reactions:
\begin{align}\label{eq:network}
\alpha_{1j}X_1 +
 \dots +
\alpha_{sj}X_s
~ \xrightarrow{\kappa_j} ~
\beta_{1j}X_1 +
 \dots +
\beta_{sj}X_s,
 \;
    {\rm for}~
	j=1,2, \ldots, m,
\end{align}
where all \defword{stoichiometric coefficients}  $\alpha_{ij}$ and $\beta_{ij}$ are non-negative integers, and we  assume that
$(\alpha_{1j},\ldots,\alpha_{sj})\neq (\beta_{1j},\ldots,\beta_{sj})$.
Each $\kappa_j \in \mathbb R_{>0}$ is called a \defword{rate constant} corresponding to the
$j$-th reaction in \eqref{eq:network}.
We say a reaction is a \defword{zero-one reaction}, if the stoichiometric coefficients $\alpha_{ij}$ and $\beta_{ij}$ in \eqref{eq:network}
belong to $\{0,1\}$.
We say a network \eqref{eq:network} is a \defword{zero-one network} if it only has zero-one reactions. 
We say a reaction is \defword{quadratic}, if the sum of stoichiometric coefficients $\alpha_{ij}$ in \eqref{eq:network} equals to $2$. 
We say a network \eqref{eq:network} is a \defword{quadratic network} if it only has quadratic reactions.
For any reaction defined in \eqref{eq:network}, the \defword{stoichiometric vector} is  $(\beta_{1j}-\alpha_{1j}, \dots, \beta_{sj}-\alpha_{sj})^\top \in \mathbb{Q}^s$. 
We say finitely many reactions are linearly independent if their corresponding stoichiometric vectors are linearly independent. 
We call the $s\times m$ matrix with
$(i, j)$-entry equal to $\beta_{ij}-\alpha_{ij}$ the
\defword{stoichiometric matrix} of
$G$, denoted by $\mathcal{N}$. 
The \defword{dimension} of $G$ is defined as the rank of $\mathcal{N}$. 
We call the $s\times m$ matrix the \defword{reactant matrix} of $G$ with $(i,j)$-entry equal to $\alpha_{ij}$, denoted by $\mathcal{Y}$.

Denote by $x_1, \ldots, x_s$ the concentrations of the species $X_1, \ldots, X_s$, respectively.
Under the assumption of mass-action kinetics, we describe how these concentrations change in time by the following system of ODEs:
{
\begin{align}\label{eq:sys}
\dot{x}\;=\;f(\kappa, x)\;:=\;\mathcal{N}v(\kappa, x),
\end{align}}
where $x=(x_1, x_2, \ldots, x_s)^\top$, $v(\kappa, x)=(v_1(\kappa, x),\dots,v_m(\kappa, x))^\top$ and
\begin{align}\label{eq:vj}
    v_j(\kappa, x)\;:=\;\kappa_j\prod \limits_{i=1}^sx_i^{\alpha_{ij}}. 
\end{align}
Treating the rate constants as a real vector of parameters $\kappa:=(\kappa_1, \dots, \kappa_m)^\top$ yields polynomials $f_{i}(\kappa,x) \in \mathbb Q(\kappa)[x]$, for $i\in\{1,\dots, s\}$. 
The polynomial ideal generated by $f$ in $\mathbb{Q}(\kappa)[x]$, denoted by  $\langle f \rangle$, is called  the \defword{steady-state ideal}.

\begin{definition}
For two $s\times m$ matrices $\mathcal{A}$ and $\hat{\mathcal{A}}$, we say the matrix $\hat{\mathcal{A}}$ \defword{has the form} of the matrix $\mathcal{A}$ if we can obtain $\hat{\mathcal{A}}$ from $\mathcal{A}$ by permuting the rows and columns.
\end{definition}
\begin{definition}
    For two networks $G$ and $\hat{G}$, we say the network $\hat{G}$ \defword{has the form} of $G$ if we can obtain $\hat{G}$ from $G$ by relabeling the species $X_1, \ldots, X_s$ as $\hat{X_1}, \ldots, \hat{X_s}$, or by relabeling the reactions ${\mathcal R}_1, \ldots, {\mathcal R}_m$ as $\hat{{\mathcal R}}_1, \ldots, \hat{{\mathcal R}}_m$.
\end{definition}
\begin{definition}\label{def:equivalent}
 Two  networks $G_1$ and $G_2$  are \defword{equivalent} if there exist $\hat{G_1}$ and $\hat{G_2}$ such that $\hat{G_i}$ has the form of $G_i$ for $i\in \{1,2\}$, and $\hat{G_1}$ and $\hat{G_2}$ have the same steady-state ideal. \end{definition}

For any given rate-constant vector $\kappa^*\in {\mathbb R}^m_{>0}$,  a \defword{steady state} 
of~\eqref{eq:sys} is a  vector of concentrations
$x^* \in \mathbb{R}_{\geq 0}^s$ such that $f(\kappa^*, x^*)=\vec{0}$, where  $f(\kappa, x)$ is on the
right-hand side of the
ODEs~\eqref{eq:sys}, and $\vec{0}$ denotes the column vector whose coordinates are all zero.
If all the coordinates of a steady state $x^*$ are strictly positive (i.e., $x^*\in \mathbb{R}_{> 0}^s$), then we call $x^*$ a \defword{positive steady state}. In applications, we usually only care about positive steady states. 
We say a  network  \defword{admits  positive steady states} (or, \defword{consistent}) if there exists a rate-constant
vector   such that  {system} \eqref{eq:sys} has at least one positive steady state. 
For the stoichiometric matrix $\mathcal{N} \in \mathbb{Q}^{s \times m}$, the \defword{positive flux cone} of $\mathcal{N}$ is defined as
\begin{align}\label{eq:fluxcone po}
\mathcal{F}^+(\mathcal{N})\; := \;\{\gamma \in {\mathbb R}_{> 0}^m|\mathcal{N} \gamma = \vec{0} \}.
\end{align}
It is well-known that a network admits positive  steady states if and only if $\mathcal{F}^+(\mathcal{N})\neq \emptyset$. 
\section{Theorem}\label{sec:thm}

\subsection{Sufficient conditions for inconsistent}\label{sub:inconsistent}
\begin{theorem}
\label{thm:coe N}
Consider a network $G$ with a stoichiometric matrix  $\mathcal{N}$. If there exists $a \in \mathbb{Q}^s$ such that $a^\top \mathcal{N}$ is non-zero and does not change sign, then the network $G$ admits no positive steady states.
\end{theorem}
\begin{proof}
Let $f$ be the steady-state system defined as in \eqref{eq:sys}. 
Let $v=(v_1, \dots, v_m)^\top$ be defined as in \eqref{eq:vj}. 
By \eqref{eq:sys}, we have 
\begin{align}
a^\top \mathcal{N} v(\kappa, x) = \sum\limits_{i=1}^s a_i f_i(\kappa, x), \label{eq:anv=af}
\end{align}
where $a_i$ denotes the coordinates of $a$. 
Assume that the network $G$ admits positive steady states. We prove the conclusion by deducing a  contradiction. 
Note that $a^\top \mathcal{N}$ is non-zero and does not change sign.
Hence, for any rate constant vector $\kappa \in \mathbb{Q}^m_{>0}$ and for any corresponding positive steady state $x \in \mathbb{Q}^s_{>0}$, the coefficients in $a^\top \mathcal{N} v(\kappa, x)$ are all positive or all negative.
By \eqref{eq:vj}, for any rate constant vector $\kappa \in \mathbb{Q}^m_{>0}$ and for any corresponding positive steady state $x \in \mathbb{Q}^s_{>0}$, we have $v(\kappa, x) \in \mathbb{Q}^m_{>0}$.
Hence, $a^\top \mathcal{N} v(\kappa, x)$ is positive or negative. 
Then, by \eqref{eq:anv=af}, for any rate constant vector $\kappa \in \mathbb{Q}^m_{>0}$ and for any positive steady state $x \in \mathbb{Q}^s_{>0}$, $\sum\limits_{i=1}^s a_i f_i(\kappa, x)$ is positive or negative.
It is contrary to $x$ is a steady state. 
Hence, the network $G$ admits no positive steady states. 
\end{proof}
\begin{remark}\label{rmk:inconsistent}
We remark that the hypotheses of Theorem \ref{thm:coe N} can be computationally checked by algebraic tools such as {\tt FindInstance} in {\tt Mathematica} \cite{mathematica2024}. 
Also note that if  a matrix ${\mathcal N
}$ has a non-zero row that  does not change sign, then the matrix satisfies the hypotheses of  Theorem \ref{thm:coe N}. See the following example. 
\end{remark}
\begin{example}
Consider the following zero-one network: 
\begin{align}
    0 \xrightarrow{\kappa_1} X_3,\;\; 0 \xrightarrow{\kappa_2}X_2, \;\;0 \xrightarrow{\kappa_3}X_2+X_3, \;\; X_1+X_2+X_3 \xrightarrow{\kappa_4}0. \nonumber
\end{align}
The stoichiometric matrix $\mathcal{N}$ is 
\begin{align}
    \begin{pmatrix}
        0&0&0&-1\\
        0&1&1&-1\\
        1&0&1&-1
    \end{pmatrix}. \nonumber
\end{align}
Let $a=(1,0,0)^\top$. Note that $a^\top \mathcal{N}=(0,0,0,-1)$. 
Hence, by Theorem \ref{thm:coe N}, the network $G$ admits no positive steady states. 
On the other hand, note that $f_1=-\kappa_4x_1x_2x_3$, and
so the steady-state system $f$ obviously has no positive solutions.  
\end{example}

\subsection{Sufficient conditions for vacuous ideals}\label{sub:gb1}
For any matrix $\mathcal{A}$, we denote by $row_i(\mathcal{A})$ and   $col_i(\mathcal{A})$ the $i$-th row  and  the $i$-th column  of $\mathcal{A}$, respectively.
For any  vector $a$, we denote by $a_i$ the $i$-th coordinate of $a$.
\begin{theorem} \label{lm:gb 1 para}
    For any  network $G$ \eqref{eq:network}, let $f$ be the steady-state system defined as in \eqref{eq:sys}. 
    If  there exists $a \in \mathbb{Q}^s$ such that 
    \begin{enumerate}
        \item[{(I)}]  
           $ P:= \mathcal{N}^\top a \neq \vec{0}$, \;and
        \item[{(II)}]  for any $i  \in \{1,\dots,s\}$, {$P_i\neq 0$ implies $col_i(\mathcal{Y}) = \vec{0}$,} 
    \end{enumerate}
  then $\langle f \rangle=\langle 1 \rangle$  in
${\mathbb Q}(\kappa)[x]$.
\end{theorem}
\begin{proof}
    Let $v=(v_1, \dots, v_m)^\top$ be defined as in \eqref{eq:vj}. 
    By {condition (I), 
    we have $P \neq \vec{0}$. 
    And, 
    we have
    \begin{align}\label{eq:gb1 pv j}
        \sum_{j=1}^n P_{i_j}v_{i_j}=P^\top v=\sum_{i=1}^s a_i row_i(\mathcal{N}) v,
    \end{align}
where $P_{i_1}, \dots, P_{i_n}$ denotes all the non-zero coordinates in $P$.}  
By {condition} {(II)} , for any $j \in \{1,\dots,n\}$, we have $col_{i_j}(\mathcal{Y})=\vec{0}$.
Hence, by \eqref{eq:vj}, for any $j \in \{1,\dots,n\}$,  
\begin{align}
    v_{i_j}=\kappa_{i_j}.\nonumber
\end{align}
Then, by \eqref{eq:sys} and by \eqref{eq:gb1 pv j}, we have 
\begin{align}
   \sum_{i=1}^s a_i f_i = \sum_{i=1}^s a_i row_i(\mathcal{N}) v=\sum\limits_{j=1}^n P_{i_j}\kappa_{i_j} \in \langle f\rangle. \nonumber
\end{align}
Thus, we have $\langle f\rangle = \langle 1\rangle $ in
${\mathbb Q}(\kappa)[x]$. 
\end{proof}
\begin{example}
We illustrate how Theorem {\ref{lm:gb 1 para}} works by the following network: 
\begin{align}
X_1 \xrightarrow{\kappa_1} X_2, \;\;X_2 \xrightarrow{\kappa_2} X_1, \;\; 0 \xrightarrow{\kappa_3} X_2. \nonumber
\end{align}
The corresponding reactant matrix is 
\begin{align}
\label{eq:ex gb1 Y}
\mathcal{Y}= \begin{pmatrix}
1 & 0 & 0\\
0 & 1 & 0
\end{pmatrix}.
\end{align}
The corresponding stoichiometric matrix is 
\begin{align}
\mathcal{N}=\begin{pmatrix}
-1 & 1 & 0\\
1 & -1 & 1
\end{pmatrix}.\nonumber
\end{align}
Let $a=(1,1)^\top$. Then,  we have $P:=\mathcal{N}^\top a= (0,0,1)^\top$. 
Note that $P_3 \neq 0$. 
By \eqref{eq:ex gb1 Y},   $col_3(\mathcal{Y}) = \vec{0}$.
 Hence, by Theorem \ref{lm:gb 1 para}, we have $\langle f\rangle = \langle 1\rangle $.
\end{example}
A row (column) of a matrix is called a \defword{non-zero row (column)} if there exists a non-zero element in this row (column). If a row (column) has both positive and negative elements, we say \defword{the row (column) changes signs}.  {Notice that if there exists one non-zero row of a stoichiometric matrix $\mathcal{N}$ that does not change sign, then the corresponding network admits no positive steady states.}
\begin{theorem}\label{lm:gb1 col}
   For any $d$-dimensional zero-one network $G$ \eqref{eq:network},  let $f$ be the steady-state system defined as in \eqref{eq:sys}. 
   We denote by $c$ the number of zero columns of the reactant matrix $\mathcal{Y}$. 
   {Assume that all the non-zero rows of the stoichiometric matrix $\mathcal{N}$ change signs.}
    Then, we have the following statements. 
    \begin{enumerate}
        \item[{(I)}] \label{lm:gb1 col c<r}There is at least one non-zero column in $\mathcal{Y}$ (i.e., $c<m$).
        \item[{(II)}] \label{lm:gb1 col c=r-1}If $c=m-1$ and $d \ge 2$, then  $\langle f \rangle=\langle 1 \rangle$  in
${\mathbb Q}(\kappa)[x]$.
        \item[{(III)}] \label{lm:gb1 col c=r-2}If $c=m-2$ and $d =s \ge 3$, then  $\langle f \rangle=\langle 1 \rangle$  in
${\mathbb Q}(\kappa)[x]$.
        \item[{(IV)}] \label{lm:gb1 col c=0}If $c=0$, then  $\langle f \rangle \neq \langle 1 \rangle$  in
${\mathbb Q}(\kappa)[x]$.
    \end{enumerate}
\end{theorem}
\begin{proof}
\begin{enumerate}
\item[{(I)}] Recall that any stoichiometric coefficients $\alpha_{ij}$ and $\beta_{ij}$ are defined as in \eqref{eq:network}.
By the definition of zero-one network, for any $(i,j) \in \{1,\dots,s\} \times \{1, \dots, m\}$, we have $\beta_{ij}\in \{0, 1\}$.
So, if there exists $(i,j)$ such that $\alpha_{ij}=0$, then the $(i,j)$-entry of $\mathcal{N}$ (i.e., $\beta_{ij}-\alpha_{ij}$) can only be $0$ or $1$. 
Hence, if all entries of $\mathcal{Y}$ are $0$ (i.e., $c=m$), then every entry of $\mathcal{N}$ is $0$ or $1$. 
By the definition of zero-one network, $\mathcal{N}$ is not a zero matrix. 
So, there exists a non-zero row of $\mathcal{N}$ whose entries are   $0$ or $1$, 
which is contrary to the assumption that all the non-zero rows of $\mathcal{N}$ change signs.
Therefore, there is at least one non-zero column in $\mathcal{Y}$ (i.e., $c<m$).
\item[{(II)}] 
Since all {the} non-zero rows of $\mathcal{N}$ change signs, there exists the element  $-1$ in every non-zero row  of $\mathcal{N}$. 
We define a new matrix $\mathcal{N}^*$ by removing the zero rows from $\mathcal{N}$. 
Also, we define $\mathcal{Y}^*$ by removing the same rows from $\mathcal{Y}$. 
By the definition of zero-one network,  we have $\alpha_{ij}\in\{0, 1\}$ and $\beta_{ij}\in\{0, 1\}$.
Hence, if there exists an $(i,j)$-entry of {$\mathcal{N}^*$} such that the entry is $-1$ (i.e., $\beta_{ij}-\alpha_{ij}=-1$), then  $\alpha_{ij}=1$. 
    Hence, there exists the element $1$ in every row  of {$\mathcal{Y}^*$}.
    {Since the number of zero columns of $\mathcal{Y}$ is $m-1$, we can assume that the first $m-1$ columns of $\mathcal{Y}$ are the zero vectors. Then, by the definition of $\mathcal{Y}^*$, the first $m-1$ columns of $\mathcal{Y}^*$ are the zero vectors.} 
    Note that there exists the element $1$ in every row of {$\mathcal{Y}^*$}. 
    So, all the coordinates of the the $m$-th column  of {$\mathcal{Y}^*$} must be $1$.    
    {Since  $d \ge 2 $ (note here, ${{\rm rank}}({\mathcal N})=d$), the matrix $\mathcal{Y}^*$ has more than one row.}
    Then, we have 
    \begin{align}
    {\mathcal{Y}^*}=
    \begin{pmatrix}
  0&  \cdots & 0 & 1\\
  & \ddots  &  & \vdots \\
  0& \cdots &  0 &1
\end{pmatrix}. \nonumber
    \end{align} 
    By the definition of zero-one network, if there exists $(i,j)$ such that $\alpha_{ij}=0$, then the $(i,j)$-entry of {$\mathcal{N}^*$} (i.e., $\beta_{ij}-\alpha_{ij}$) can only be $0$ or $1$. 
    Hence, every coordinate of the first $m-1$ columns of {$\mathcal{N}^*$} can only be $0$ or $1$. 
   Note that there exists $-1$ in every row vector of {$\mathcal{N}^*$}. So, all the coordinates of the the $m$-th column  of {$\mathcal{N}^*$} are $-1$. 
   Thus, {$\mathcal{N}^*$} has the following form
       \begin{align}\label{eq:nform}
    {\mathcal{N}^*}=
    \begin{pmatrix}
  &   &  & -1\\
  & \mathcal{A}_1  &  & \vdots \\
  &  &   &-1
\end{pmatrix},
    \end{align}
{where all the entries of the  submatrix $\mathcal{A}_1$ are $0$ or $1$. 
Recall that for any $i \in \{1,\dots,s\}$, $row_i(\mathcal{N})$ denotes the $i$-th row  of $\mathcal{N}$.
Since  $d \ge 2 $, there exist two non-zero rows $row_i(\mathcal{N})$ and $row_j(\mathcal{N})$} such that $P^\top:=row_i(\mathcal{N})-row_j(\mathcal{N})\neq \vec{0}^\top$.
So,  if $P_k\neq 0$, then by \eqref{eq:nform}, we have $k \in \{1,\dots, m-1\}$.
Note that for any $k \in \{1,\dots, m-1\}$, the $k$-th column vector of $\mathcal{Y}$ is a zero vector. 
Therefore, by Theorem \ref{lm:gb 1 para}, we have $\langle f\rangle=\langle 1\rangle$. 
\item[{(III)}] Assume that the first $m-2$ columns of $\mathcal{Y}$ are the zero vector.
    Then, by the definition of zero-one network, we have
  \begin{align}
    \mathcal{Y}=
    \begin{pmatrix}
  0&  \cdots & 0 & \\
  & \ddots  &  & \mathcal{A}_2 \\
  0& \cdots &  0 &
\end{pmatrix}, \nonumber
    \end{align}
    where the entries of the $s \times 2$ submatrix $\mathcal{A}_2$ are $0$ or $1$. 
    {Since $d=s$, all rows of $\mathcal{N}$ are non-zero.} 
    Then,  similar to the proof of (2), by the definition of zero-one network, we have 
      \begin{align}
      \label{eq:gb1 m-2 N}
    \mathcal{N}=
    \begin{pmatrix}
 \mathcal{A}_3 & \mathcal{A}_4
\end{pmatrix},
    \end{align}
    where the entries of the $s \times (m-2)$ submatrix $\mathcal{A}_3$ are $0$ or $1$.
    {Since all rows of $\mathcal{N}$ change signs,} there exists the element $-1$ in every row  of $\mathcal{N}$.
    Hence, the rows of the $s \times 2$ submatrix   $\mathcal{A}_4$ must be in the set
    \begin{align}\label{eq:row B kind} 
    \{(0,-1),(-1,0),(1,-1),(-1,1),(-1,-1)\}.
    \end{align}
    \begin{itemize}
    \item[(i)] Assume that there exist $k<s$ different row vectors in \eqref{eq:row B kind} appearing in $\mathcal{A}_4$.
       Since $k<s$,   at least two  rows  in $\mathcal{A}_4$ are the same. 
        Assume that the $i$-th and the $j$-th ($i, j \in \{1,\dots,s\}$) rows in $\mathcal{A}_4$ are the same.
        Define $P^\top:=row_i(\mathcal{N})-row_j(\mathcal{N})$. Then, we have 
        $P_{m-1}=P_{m}=0$.
        Since  ${{\rm rank}}(\mathcal N)=d=s$, we have $P\neq \vec{0}$.  
        Therefore, for any $k \in \{1,\dots,m\}$, if $P_k\neq 0$, we have $k \in \{1,\dots,m-2\}$. 
        Note that for any $k \in \{1,\dots,m-2\}$, the $k$-th column of $\mathcal{Y}$ is the zero vector. 
         Hence, by Theorem \ref{lm:gb 1 para}, we have $\langle f\rangle=\langle 1\rangle$. 
        \item[(ii)] Assume that  there exist $s$ different row vectors in \eqref{eq:row B kind} appearing in $\mathcal{A}_4$. 
        Note that any matrix consisting of $s\ge3$ different row vectors in \eqref{eq:row B kind} is rank-two.
        Hence, ${\rm rank}(\mathcal{A}_4)=2$. 
        Note that there are $s\ge3$ rows in $\mathcal{A}_4$. 
        So, there exists $a\in \mathbb{Q}^s\setminus\{\vec{0}\}$ such that $\mathcal{A}_4^\top a=\vec{0}$. 
        Recall that $a=(a_1, \dots, a_s)^\top$. 
        By \eqref{eq:gb1 m-2 N},  $\mathcal{A}_4^\top a$ is equal to the last two coordinates of $P^\top:=\sum\limits_{i=1}^s a_i row_i(\mathcal{N})$. 
        Hence, we have $P_{m-1}=P_m=0$. 
        Note that ${\rm rank} (\mathcal{N})=d=s$.
        So, $P \neq \vec{0}$ since $a\neq \vec{0}$. 
        Hence, for any $k \in \{1,\dots,m\}$, if $P_k\neq 0$, then we have $k \in \{1,\dots,m-2\}$. 
        Note that for any $k \in \{1,\dots,m-2\}$, the $k$-th column of $\mathcal{Y}$ is the zero vector. 
         Therefore, by Theorem \ref{lm:gb 1 para}, we have $\langle f\rangle=\langle 1\rangle$.
    \end{itemize}
    \item[{(IV)}] Since there is no zero columns in $\mathcal{Y}$, by \eqref{eq:vj}, for any $i \in \{1,\dots,m\}$, $v_i = \kappa_i \prod_{j=1}^s x^{\alpha_{ji}}$, where {$\sum\limits_{j=1}^s \alpha_{ji} > 0$} since $\alpha_{ji}\in \{0,1\}$. 
    {Hence, by \eqref{eq:sys}, for any $i \in \{1,\dots,s\}$, if $row_i(\mathcal{N})$ is a zero vector, then $f_i=0$. 
    By \eqref{eq:sys}, for any $i \in \{1,\dots,s\}$, if $row_i(\mathcal{N})$ is non-zero,  then all terms in  $f_i$ contain the variables $x_j$'s. }
    So, in ${\mathbb Q}(\kappa)[x]$, $1 \notin \langle f \rangle$. Hence, we have $\langle f \rangle \neq \langle 1 \rangle$. 
\end{enumerate}
\end{proof}
\begin{example}
We illustrate how  Theorem \ref{lm:gb1 col} works by the following examples. \\
 {(I)}   Consider the following zero-one network: 
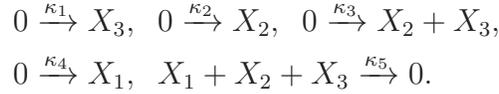
\begin{align}
    &0 \xrightarrow{\kappa_1} X_3,\;\; 0 \xrightarrow{\kappa_2}X_2, \;\;0 \xrightarrow{\kappa_3}X_2+X_3, \nonumber\\ &0 \xrightarrow{\kappa_4}X_1,\;\; X_1+X_2+X_3 \xrightarrow{\kappa_5}0. \nonumber
\end{align}
The stoichiometric matrix $\mathcal{N}$ is 
\begin{align}
    \begin{pmatrix}
        0&0&0&1&-1\\
        0&1&1&0&-1\\
        1&0&1&0&-1
    \end{pmatrix}. \nonumber
\end{align}
Note that ${\rm rank}(\mathcal{N})=3$. 
The reactant matrix $\mathcal{Y}$ is 
\begin{align}
\begin{pmatrix}
      0&0&0&0&1\\
    0&0&0&0&1\\
    0&0&0&0&1  
\end{pmatrix}. \nonumber
\end{align}
Note that there are $4$  zero columns in $\mathcal{Y}$ (i.e., $c=4=m-1$). 
Hence, by Theorem \ref{lm:gb1 col} {(II)}, we have $\langle f \rangle=\langle 1 \rangle$. 
Note that
\begin{align}
    f_1=&\kappa_4- \kappa_5x_1x_2x_3, \nonumber\\
    f_2=&\kappa_2 +\kappa_3- \kappa_5x_1x_2x_3, \nonumber\\
    f_3=&\kappa_1 +\kappa_3- \kappa_5x_1x_2x_3.\nonumber
\end{align}
It is straightforward to check that  $\frac{f_1-f_2}{\kappa_4-\kappa_2-\kappa_3}=1$. 

{(II)}   Consider the following zero-one network: 
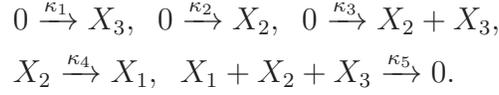
\begin{align}
    &0 \xrightarrow{\kappa_1} X_3,\;\; 0 \xrightarrow{\kappa_2}X_2, \;\;0 \xrightarrow{\kappa_3}X_2+X_3,\nonumber\\ &X_2 \xrightarrow{\kappa_4}X_1,\;\; X_1+X_2+X_3 \xrightarrow{\kappa_5}0.\nonumber
\end{align}
The stoichiometric matrix $\mathcal{N}$ is 
\begin{align}
    \begin{pmatrix}
        0&0&0&1&-1\\
        0&1&1&-1&-1\\
        1&0&1&0&-1
    \end{pmatrix}.\nonumber
\end{align}
Note that ${\rm rank}(\mathcal{N})=3$, i.e., $d=s =3$. 
The reactant matrix $\mathcal{Y}$ is 
\begin{align}
\begin{pmatrix}
      0&0&0&0&1\\
    0&0&0&1&1\\
    0&0&0&0&1  
\end{pmatrix}.\nonumber
\end{align}
Note that there are $3$  zero columns in $\mathcal{Y}$ (i.e., $c=3=m-2$). 
Hence, by Theorem \ref{lm:gb1 col} {(III)}, we have $\langle f \rangle=\langle 1 \rangle$. 
Note that
\begin{align}
    f_1=&\kappa_4x_2- \kappa_5x_1x_2x_3, \nonumber\\
    f_2=&\kappa_2 +\kappa_3-\kappa_4x_2- \kappa_5x_1x_2x_3,\nonumber \\
    f_3=&\kappa_1 +\kappa_3- \kappa_5x_1x_2x_3.\nonumber
\end{align}
It is straightforward to check that  $\frac{f_1+f_2-2f_3}{-2\kappa_1+\kappa_2-\kappa_3}=1$. 

{(III)}    Consider the following zero-one network: 
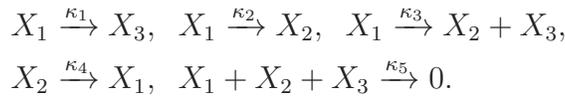
\begin{align}
    &X_1 \xrightarrow{\kappa_1} X_3,\;\; X_1 \xrightarrow{\kappa_2}X_2, \;\;X_1 \xrightarrow{\kappa_3}X_2+X_3,\nonumber\\ &X_2 \xrightarrow{\kappa_4}X_1,\;\; X_1+X_2+X_3 \xrightarrow{\kappa_5}0. \nonumber
\end{align}
The stoichiometric matrix $\mathcal{N}$ is 
\begin{align}
    \begin{pmatrix}
        -1&-1&-1&1&-1\\
        0&1&1&-1&-1\\
        1&0&1&0&-1
    \end{pmatrix}. \nonumber
\end{align}
The reactant matrix $\mathcal{Y}$ is 
\begin{align}
\begin{pmatrix}
    1&1&1&0&1\\
    0&0&0&1&1\\
    0&0&0&0&1  
\end{pmatrix}. \nonumber
\end{align}
Note that there is no  zero columns in $\mathcal{Y}$ (i.e., $c=0$). 
Hence, by Theorem \ref{lm:gb1 col} {(IV)}, we have $\langle f \rangle \neq \langle 1 \rangle$. 
Note that
\begin{align}
    f_1=&-\kappa_1x_1-\kappa_2x_1-\kappa_3x_1+\kappa_4x_2- \kappa_5x_1x_2x_3, \nonumber\\
    f_2=&\kappa_2x_1 +\kappa_3x_1-\kappa_4x_2- \kappa_5x_1x_2x_3, \nonumber\\
    f_3=&\kappa_1x_1 +\kappa_3x_1- \kappa_5x_1x_2x_3. \nonumber
\end{align}
It is straightforward to check that  $1 \notin \langle f \rangle$.
\end{example}
\begin{remark}
We remark that although Theorem \ref{lm:gb 1 para}  is straightforward to prove, the hypotheses in the theorems  {cannot} be directly checked by a computer, while Theorem \ref{lm:gb1 col} is explicit enough.   In Section \ref{sec:alg}, {we solve} this issue by {Algorithm \ref{alg:gb1}}.   
\end{remark}

\subsection{Sufficient conditions for equivalence}\label{sub:equ}
\begin{lemma}\label{lm:ideal equal}
   For any  two  networks $G$ and $\hat{G}$ {with $s$ species and $m$ reactions}, let $\mathcal{N}$ and $\hat{\mathcal{N}}$ be the corresponding stoichiometric matrices,  and let $\mathcal{Y}$ and $\hat{\mathcal{Y}}$ be the reactant matrices. 
    Let $f$ and $\hat{f}$ be the steady-state systems defined as in \eqref{eq:sys}. 
    If  \begin{enumerate}
        \item[{(I)}] $\mathcal{Y}=\hat{\mathcal{Y}}$, and
        \item[{(II)}] for any $ i \in \{1, \dots, s\}$, we have
        $$row_i(\mathcal{N}) \in \operatorname{span}_{\mathbb{Q}} \{row_1(\hat{\mathcal{N}}), \dots, row_s(\hat{\mathcal{N}})\},$$ and $$row_i(\hat{\mathcal{N}}) \in \operatorname{span}_{\mathbb{Q}} \{row_1(\mathcal{N}), \dots, row_s(\mathcal{N})\},$$ 
    \end{enumerate}
    where $\operatorname{span}_{\mathbb{Q}}\{\cdot\}$ means the rational vector space spanned by a set of rational vectors,  
    then $\langle f \rangle= \langle \hat{f} \rangle$ in ${\mathbb Q}(\kappa)[x]$.
\end{lemma}
\begin{proof}
Let $v=(v_1,\dots,v_m)^\top$ and $\hat{v}=(\hat{v}_1,\dots,\hat{v}_m)^\top$ be defined as in \eqref{eq:vj} corresponding to $G$ and $\hat{G}$ respectively, where 
\begin{align}
\label{eq:vi vi'}
 v_i = \kappa_i \prod_{j=1}^s x_j^{\alpha_{ji}},\;\;   \hat{v}_i = \kappa_i \prod_{j=1}^s x_j^{\hat{\alpha}_{ji}}.
\end{align}
By $\mathcal{Y}=\hat{\mathcal{Y}}$, for any $i\in\{1,\dots,s\}$ and for any $j \in \{1,\dots,s\}$, $\alpha_{ij}=\hat{\alpha}_{ij}$. 
Hence, by \eqref{eq:vi vi'},  we have $v=\hat{v}$.  So, the conclusion follows from
\eqref{eq:sys} and {condition} {(II)} .  
\end{proof}
\begin{remark}
If $G$ has less reactions than $\hat{G}$ does, then we can 
check if there exists a subnetwork of  $\hat{G}$ such that  the  subnetwork is equivalent {to} $G$. Here, by ``subnetwork", we mean the network consists of some reactions of  $\hat{G}$ and it has the same stoichiometric subspace as $\hat{G}$, where stoichiometric subspace means the column space of  {the} stoichiometric matrix. 
By \cite{BP16, MBJ2023}, if a subnetwork admits multistability or Hopf bifurcations, then the original network also admits  multistability or Hopf bifurcations.
\end{remark}

\begin{example}
Consider the following two networks $G$ and $\hat{G}$.
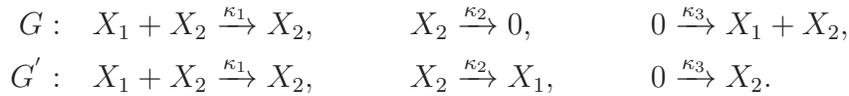
\begin{align}
G:& \quad X_1+X_2 \xrightarrow{\kappa_1} X_2, &&X_2 \xrightarrow{\kappa_2} 0, &&0  \xrightarrow{\kappa_3} X_1+X_2,\nonumber\\
\hat{G}:&\quad  X_1+X_2 \xrightarrow{\kappa_1} X_2, &&X_2 \xrightarrow{\kappa_2} X_1, &&0 \xrightarrow{\kappa_3} X_2.\nonumber
\end{align}
Note that the two networks $G$ and $\hat{G}$ have the same set of reactants.  Hence, the reactant matrices $\mathcal{Y}$ and $\hat{\mathcal{Y}}$ are equal.
The stoichiometric matrices  $\mathcal{N}$ and $\hat{\mathcal{N}}$ are as follows.
\begin{align}
\mathcal{N}=\begin{pmatrix}
-1& 0&1\\
0 &-1 & 1
\end{pmatrix}, \quad 
\hat{\mathcal{N}}=\begin{pmatrix}
-1 & 1 &0\\
0 & -1 & 1
\end{pmatrix}. \nonumber
\end{align}
Note that 
\begin{align*}
    row_1(\mathcal{N}) &= row_1(\hat{\mathcal{N}}) +row_2(\hat{\mathcal{N}}),\\
    row_2(\mathcal{N}) &= row_2(\hat{\mathcal{N}}),\\
    row_1(\hat{\mathcal{N}}) &= row_1(\mathcal{N}) -row_2(\mathcal{N}).
\end{align*}
So, by {Lemma} \ref{lm:ideal equal}, the ideals generated by $f$ and $\hat{f}$ are equal.
\end{example}
\begin{remark}
 Notice that each network is uniquely  determined by a pair of stoichiometric matrix $\mathcal{N}$ and reactant matrix $\mathcal{Y}$.
 However, a stoichiometric matrix $\mathcal{N}$ might correspond to many different reactant matrices. 
 For example, if a network has the following stoichiometric matrix 
    \begin{align}
    \label{eq:NY4 N}
        \mathcal{N}=
        \begin{pmatrix}
            1 & 1 & 0 &-1 \\
            1&-1&-1&1
        \end{pmatrix},
    \end{align}
    then the  reactant matrix can be one of the following two matrices  
     \begin{align}
            {\mathcal Y}_1=\begin{pmatrix}
             0&0&0&1\\
             0&1&1&0
            \end{pmatrix},\quad
            {\mathcal Y}_2= \begin{pmatrix}
             0&0&1&1\\
             0&1&1&0
            \end{pmatrix}. \nonumber
        \end{align}
 So, for two networks, if  their stoichiometric matrices  $\mathcal{N}$ and $\hat{\mathcal{N}}$  have different row spaces, then  we do not need to compare the reactant matrices  $\mathcal{Y}$ and $\hat{\mathcal{Y}}$. That means in practice, we can first check condition (II) of Lemma \ref{lm:ideal equal}, and if it holds, we check condition (I). In Section \ref{sec:alg}, we  check if condition (II)  holds by checking if the reduced row echelon forms of the matrices $\mathcal{N}$ and $\hat{\mathcal{N}}$ are equal. See Step 4 in Algorithm \ref{alg:main}. 
\end{remark}
\section{Algorithm}\label{sec:alg}
In this section, we propose
Algorithm \ref{alg:main}, together with two sub-algorithms (Algorithm \ref{alg:inconsistent} and Algorithm \ref{alg:gb1}),  to determine the equivalence of zero-one networks.  The goal is to efficiently  classify  all  $d$-dimensional  zero-one networks consisting of $s$ species $m$ reactions according to the equivalence determined by their steady-state ideals.  The correctness of Algorithm \ref{alg:main} follows from {Theorems \ref{thm:coe N}--\ref{lm:gb1 col}}, {Lemma} \ref{lm:ideal equal} and  the correctness of the {sub-algorithms}.

\begin{algorithm}\label{alg:main}
DetermineEquivalence
    \begin{description}
        \item[Input.] The number of {species} $s$, the number of {reactions} $m$ and the dimension $d$
        \item[Output.] Finitely many classes of the steady-state systems $f$ corresponding to the $d$-dimensional $s$-species $m$-reaction zero-one  networks such that the systems in the same class generate the same steady-state ideal
        \item[Step 1.] Enumerate the
        stoichiometric coefficients $\alpha_{ij}$ and $\beta_{ij}$ for all $s$-species zero-one reactions and compute the corresponding stoichiometric vectors $\beta-\alpha$. Then,  enumerate all possible $s \times m$ stoichiometric matrices $\mathcal{N}$ with rank $d$ whose columns  are formed by these stoichiometric vectors.  
        \item[Step 2.]  We call CoefficientsForInconsistent($s$, $d$, $m$), and we get a set $S_1 \subset \mathbb{Q}^s$. For each stoichiometric  matrix obtained in Step 1, we do the following steps. 
        \begin{description}
            \item[{Step 2.1.}]
        If there exists $a\in S_1$ such that the vector $a^\top \mathcal{N}$ is non-zero and does not change sign, then delete this matrix.
         \item[{Step 2.2.}]
        If the set $\mathcal{F}^+(\mathcal{N})$ is empty, then delete this matrix. 
        \end{description}
        ($\#$By Step 2, we remove the stoichiometric matrices whose all possible corresponding networks definitely admit no positive steady states by Theorem \ref{thm:coe N} and a standard way.)
        \item[{Step 3.}] For each remaining stoichiometric matrix,  delete the stoichiometric matrices that have the same form with it. 
        \item[{Step 4.}] For every remaining stoichiometric matrix $\mathcal{N}$, generate all possible corresponding reactant matrices $\mathcal{Y}$ for zero-one networks. For every pair $\{\mathcal{N}, \mathcal{Y}\}$, we do the following steps.
                 \begin{description}
            \item[{Step 4.1.}] If
        \begin{enumerate}
            \item[(I)] $c=m-1$, and $d \ge 2$, or
            \item[(II)] $c=m-2$, and $d =s \ge 3$,
        \end{enumerate}
        then delete the pair. ($\#$Note here, $c$ denotes the number of zero columns of the reactant matrix $\mathcal{Y}$.) 
        \item[{Step 4.2.}] We call CoefficientsForVacuous($s$, $d$), and we get a set $S_2 \subset \mathbb{Q}^s$.
        If there exists $a\in S_2$ such that
        \begin{enumerate}
            \item[{(I)}] $\mathcal{N}^\top a \neq \vec{0}$, and
            \item[{(II)}] for any $i \in \{1,\dots,s\}$, the $i$-th coordinate of $\mathcal{N}^\top a$ is non-zero implies $col_i(\mathcal{Y})= \vec{0}$,
        \end{enumerate}
        then delete the pair. 
        \end{description}
        ($\#$By Step 4,  we remove the networks that have vacuous steady-state ideals according to Theorems \ref{lm:gb 1 para}--\ref{lm:gb1 col}.)
        \item[{Step 5.}] 
        Combine the remaining pairs $\{\mathcal{N}, \mathcal{Y}\}$ into the same group if they have  the same reduced row echelon form of $\mathcal{N}$. For each group, we do the following steps.    
        \begin{description}
            \item[{Step 5.1.}] For each pair $\{\mathcal{N}, \mathcal{Y}\}$ and $\{\hat{\mathcal{N}}, \hat{\mathcal{Y}}\}$, we do the following steps. 
            \begin{description}
                \item[{Step 5.1.1.}] Check if there exists a permutation on the rows of $\mathcal{Y}$ such that $\mathcal{Y}$ and $\hat{\mathcal{Y}}$ are equal after permuting the rows of $\mathcal{Y}$.   
          \begin{description}
            \item[{Step 5.1.1.1.}]  If yes, then put the pair of matrices $\{\mathcal{N}, \mathcal{Y}\}$ and $\{\hat{\mathcal{N}}, \hat{\mathcal{Y}}\}$ into the same group. 
                \item[{Step 5.1.1.2.}] If not, then consider the next pair of matrices.
        \end{description}
         \end{description}
          ($\#$By Step 5.1, we check the equivalence according to Lemma \ref{lm:ideal equal}.)
        \item[{Step 5.2.}] For every group, choose the first pair as the representative element.  
        Note that every pair of matrices $\{\mathcal{N}, \mathcal{Y}\}$ uniquely determines a network. 
        \end{description}
      \item[{Step 6.}]  For every representative network obtained from Step 5, calculate the reduced Gr\"obner bases for steady-state ideal $\langle f \rangle$.  Delete the systems  whose corresponding reduced Gr\"obner bases are $\{1\}$ or have a monomial as an element.\\
      {($\#$Note here, if the Gr\"obner basis contains a single monomial, then the network corresponding to the system admits no positive steady states for a generic choice of rate constants.)}
      \item[{Step 7.}] Combine the networks with the same reduced Gr\"obner basis into the same class. Choose the first network as the representative element. Return these representative elements.
    \end{description}
\end{algorithm}
\begin{remark}

    Note that in Step 1, we first enumerate the stoichiometric coefficients $\alpha_{ij}$ and $\beta_{ij}$, rather than the stoichiometric vectors directly. This approach facilitates extensions of the algorithm, particularly for the quadratic zero-one network where we require that the sum of $\alpha_{ij}$ for every reactant equals $2$.
\end{remark}
\begin{remark}

   We remark that Step 2.1 in Algorithm \ref{alg:main} generalizes  \cite[Algorithm 1--Step 2.1]{YX2025}, which follows from Remark 
\ref{rmk:inconsistent} and the correctness of the following Algorithm \ref{alg:inconsistent}. 
\end{remark}
\begin{algorithm}\label{alg:inconsistent}
CoefficientsForInconsistent
\begin{description}
\item [Input.] The number of species $s$ and the number of the reactions $m$
\item [Output.] A set $S \subset \mathbb{Q}^s$ such that for any  $m$-reaction $s$-species zero-one network, if there exists $a \in \mathbb{Q}^s$ such that the condition of Theorem \ref{thm:coe N} holds, then there exists $\tilde{a}\in S$ satisfying  the condition of Theorem \ref{thm:coe N}.
    \item [Step 1.] Enumerate all possible $s$-species zero-one stoichiometric vectors.
    \item [Step 2.]
    Let $S=\{e_1,\dots, e_s\}$, where $e_i \in \mathbb{Q}^s$ satisfies the $i$-th coordinate is $1$ and other coordinates are $0$.
    Let $w:=(w_1, \ldots, w_s)$ be an indeterminate vector in  $ \mathbb{Q}_{\ge 0}^s$.  Enumerate all relations among $w_1, \ldots, w_s$.    For every relation, we do the following steps.
        \begin{description}
        \item[Step 2.1.]  For each stoichiometric vector $\gamma$ generated in Step 1, if it satisfies $w^\top \gamma < 0$, then delete it. If not, then keep it. 
        \item[Step 2.2.]       
        For the remaining stoichiometric vectors in Step 2.1,
        enumerate all possible choices of $m$ vectors from them, forming a stoichiometric matrix $\mathcal{N}$.  For each $\mathcal{N}$, we do the following step.
        \begin{description}
        \item[Step 2.2.1.]  
       If all non-zero rows of $\mathcal{N}$ change signs, then check  if there exists  $q\in \mathbb{Q}^s$ such that $q^\top \mathcal{N}$ is non-zero and does not change sign. ($\#$Here, we use {\tt FindInstance} in {\tt Mathematica}.) If so, add the Cartesian product $\bigotimes\limits_{i=1}^s\{q_i, -q_i\}$ to the set $S$.
        \end{description}
          \end{description}
\end{description}
\end{algorithm}

\textbf{Proof of the Correctness of Algorithm \ref{alg:inconsistent}}
For any $m$-reaction $s$-species zero-one network with the stoichiometric matrix $\mathcal{N}$, assume that there exists $a \in \mathbb{Q}^s$ such that the condition of Theorem \ref{thm:coe N} holds, i.e., $a^\top \mathcal{N}$ is non-zero and does not change sign.  Note that if there exists $i \in \{1,\dots,s\}$ such that ${\rm row}_i({\mathcal N})$ is non-zero and does not change sign, then $e_i^\top \mathcal{N} = {\rm row}_i({\mathcal N})$ is non-zero and does not change sign. Therefore, by Step 2,  the vector $\tilde{a} = e_i \in S$ satisfies the condition of Theorem \ref{thm:coe N}. Below, we assume that all non-zero rows of ${\mathcal N}$ change signs. 
\begin{itemize}
    \item[(I)] Assume that all entries of $a^\top \mathcal{N}$ are non-negative. 
    Let $\hat{a} = |a| \in \mathbb{Q}_{\ge 0}^s$. Then, there exists a relation of $w$ in Step 2 that satisfies the same relation of $\hat{a}$.  
    Let $\hat{\mathcal N}$ be the matrix formed by the rows ${\rm sign}(a_i){\rm row}_i({\mathcal N})$ $(i=1, \ldots, s)$.   Note that 
\begin{align*}
    a^\top \mathcal{N} =\hat{a}^\top\hat{\mathcal{N}}.
\end{align*}

Since $a^\top \mathcal{N}$ is non-zero and all its entries are non-negative, the same holds for $\hat{a}^\top\hat{\mathcal{N}}$. 
Hence, for every column vector $\gamma$ of the matrix $\hat{\mathcal{N}}$, we have $\hat{a}^\top \gamma \ge 0$. 
Then, by Step 2.1 and Step 2.2, $\hat{\mathcal{N}}$ is one of the enumerated matrices under the relation of $\hat{a}$. 
Since all non-zero rows of $\mathcal{N}$ change signs, then every non-zero row of $\hat{\mathcal{N}}$ also changes sign. 
Then, by Step 2.2.1.2, the vector $q \in \mathbb{Q}^S$ obtained via {\tt FindInstance} satisfies that $q^\top \hat{\mathcal{N}}$ is non-zero and all entries are non-negative.  
Let $\tilde{a}=
({\rm sign}(a_1) q_1, \ldots, {\rm sign}(a_s) q_s)$. By Step 2.2.1.2, we know that $\tilde{a} \in S$. Notice that $\tilde{a}^\top \mathcal{N} = q^\top \hat{\mathcal{N}}$. Hence, $\tilde{a}^\top \mathcal{N}$ is non-zero and all its entries are non-negative, i.e., $\tilde{a}$ satisfies the condition of Theorem \ref{thm:coe N}.

\item[(II)] Assume that all entries of $a^\top \mathcal{N}$ are non-positive. Notice that $-a^\top \mathcal{N}$ is non-zero and all entries of $-a^\top \mathcal{N}$ are non-negative. Then, the proof is similar to (I).  
\end{itemize}

\begin{example}\label{ex:inconsistent}
We illustrate how Algorithm \ref{alg:inconsistent} works for $m=5$ and $s=3$.
\begin{description}
    \item [Step 1.] Enumerate all possible $3$-species  stoichiometric vectors according to zero-one reactions, i.e., all vectors in 
    \begin{align*}
    \{-1,0,1\}^3 \setminus (0,0,0)^\top.
    \end{align*}
 \item [Step 2.] Let $S = \{(0,0,1)^\top, (0,1,0)^\top, (0,0,1)^\top\}$. Enumerate all relations among $w_1$, $w_2$, and $w_3$, where $w_i \in \mathbb{Q}_{\ge 0}$. By symmetry, we only need to consider the following cases. 
    \begin{align*}
    w_1>w_2>w_3, \ w_1=w_2>w_3, \ w_1>w_2=w_3, \ w_1=w_2=w_3.  
    \end{align*} We treat  every relation as follows. 
    \begin{itemize}
        \item[(I)] $w_1>w_2>w_3$. 
        \begin{description}
        \item[Step 2.1.]  For each stoichiometric vector $\gamma$ generated in Step 1, if it satisfies $w^\top \gamma < 0$, then delete it. If not, then keep it. For example, 
        \begin{itemize}
            \item[(i)] we delete $(-1,-1,-1)^\top$ since $w^\top (-1,-1,-1)^\top = -w_1-w_2-w_3 <0$;
            \item[(ii)] we keep 
        $(-1,1,1)^\top$ since we cannot determine the sign of $w^\top(-1,1,1)^\top = -w_1+w_2+w_3$ under the condition $w_1>w_2>w_3 \ge 0$. 
        \end{itemize} 
        The remaining vectors are 
        \begin{align*}
             &(-1,1,1)^\top, &&(0,0,1)^\top, &&(0,1,-1)^\top, &&(0,1,0)^\top, &&(0,1,1)^\top, \nonumber \\
             &(1,-1,-1)^\top, &&(1,-1,0)^\top, &&(1,-1,1)^\top, &&(1,0,-1)^\top, &&(1,0,0)^\top, \nonumber \\
             &(1,0,1)^\top, &&(1,1,-1)^\top, &&(1,1,0)^\top, &&(1,1,1)^\top.
        \end{align*}
        \item[Step 2.2.]  
       For the remaining stoichiometric vectors in Step 2.1,
        enumerate all possible choices of $5$ vectors from them, forming a stoichiometric matrix $\mathcal{N}$.  For each $\mathcal{N}$, we do Step 2.2.1 of Algorithm \ref{alg:inconsistent}. 
        For example, consider  
            $$\mathcal{N} = \begin{pmatrix}
            -1&0&0 & 0 &1\\
            1&0&1&1&-1\\
            1&1&-1&0&-1
        \end{pmatrix}.$$ \begin{description}
        \item[Step 2.2.1.]  
      Notice that all non-zero rows of $\mathcal{N}$ change signs. Using {\tt FindInstance} in {\tt Mathemetica}, we find $q = (2, 1, 1)^\top$ such that $q^\top \mathcal{N}$ is non-zero and does not change sign. Then, add the Cartesian product $\{2,-2\}, \times\{1,-1\}, \times\{1,-1\}$ to set $S$. 
       \end{description} 
        \end{description}
        After enumerating all possible matrices, we add the follow set of vectors into $S$: $$
          \{0, \pm\frac{1}{4}, \pm\frac{1}{3}, \pm\frac{1}{2}, \pm\frac{2}{3}, \pm\frac{3}{4}, \pm 1, \pm\frac{3}{2},  \pm2, \pm 3\}^3.
           $$      
           \item[(II)] We implement the similar steps as case (I) for the other relations $w_1=w_2>w_3$, $w_1>w_2=w_3$ and $w_1=w_2=w_3$. And finally, we get 
           \begin{align}\label{eq:inconsistent}
           S = \{0, \pm\frac{1}{4}, \pm\frac{1}{3}, \pm\frac{1}{2}, \pm\frac{2}{3}, \pm\frac{3}{4}, \pm 1, \pm\frac{3}{2},  \pm2, \pm 3\}^3.
           \end{align}
              \end{itemize}
\end{description}

\end{example}

\begin{algorithm}\label{alg:gb1} 
CoefficientsForVacuous
    \begin{description}
    \item [Input.] The number of species $s$ and the dimension $d$
  \item [Output.] A set $S \subset \mathbb{Q}^s$ such that for any  $d$-dimensional $s$-species zero-one network, if there exists $a \in \mathbb{Q}^s$ such that {conditions} {(I)} and {(II)} of Theorem \ref{lm:gb 1 para} hold, 
then there exists $\tilde{a}\in S$ satisfying  {conditions} {(I)} and {(II)} of Theorem \ref{lm:gb 1 para}
    \item [Step 1.] Enumerate all $s$-species zero-one reactions.     
    \item [Step 2.] Let $S:=\emptyset$. For every $n \in \{1,\dots, d-1\}$,  {for all possible choices of  $n$ linearly independent reactions from the reactions enumerated in Step 1, we do the following steps. Note that here we operate on all possible combinations of $n$ reactions.} 
    \begin{description}
        \item[Step 2.1.] Denote by $$\mathcal{M}=\begin{pmatrix}
            M_1\\
            \vdots\\
            M_n
        \end{pmatrix}$$ an $n\times s$ matrix, where each row $M_j$  is the transpose of the stoichiometric vector of a reaction from the chosen $n$ independent reactions. 
        \item[Step 2.2.] {Calculate a basis of the nullspace of $\mathcal{M}$ in $\mathbb{Q}^s$} and add all the vectors  in the basis into the set $S$.
        
    \end{description}
\item[Step 3.] Return the set $S$. 
    \end{description}
\end{algorithm}

{\bf Proof of the Correctness of Algorithm \ref{alg:gb1}.}
{For any  $d$-dimensional $s$-species zero-one network, suppose ${\mathcal Y}$ and ${\mathcal N}$ are the
reactant matrix and the stoichiometric matrix.} Define 
\begin{align}
\label{eq:delta1}
    \Delta_1&:=\{i \in \{1,\dots,m\}| col_i(\mathcal{Y}) \neq \vec{0}\}, \\
    \label{eq:delta2}
    \Delta_2&:=\{1,\dots,m\} \setminus \Delta_1. 
\end{align}
Let $M$ and $M_k$ ($k\in \{1,2\}$) be the real vector spaces spanned by
the columns of ${\mathcal N}$
and the set of column vectors 
$\{col_i({\mathcal N})|i\in \Delta_k\}$, respectively. 
{If there exists $a \in \mathbb{Q}^s$ such that {conditions} {(I)} and {(II)} of Theorem \ref{lm:gb 1 para} hold, then}
by \eqref{eq:delta1} and  {condition} {(II)} of Theorem \ref{lm:gb 1 para}, we have   \begin{align*}
    a\in M_1^\perp. 
\end{align*} 
Since $\mathcal{N}^\top a\neq \vec{0}$ ({condition} {(I)}  of Theorem \ref{lm:gb 1 para}), we have $a\not\in M^\perp$. Since  
$M_1\subset M$,  we have  ${{\rm dim}}(M_1)<{{\rm dim}}(M)=d$. So, by Step 2 of Algorithm \ref{alg:gb1}, 
there exists $B\subset S$ such that $B$ is {a} basis of  {$M^{\perp}_1$}.
Notice that there exists $\tilde{a} \in B$ such that
 $\tilde{a}\not\in M^\perp$, i.e., 
 $\mathcal{N}^\top \tilde{a}\neq\vec{0}$. Otherwise, 
 we have $B\subset M^\perp$ and hence, we have $M_1^\perp\subset M^\perp$, which is contrary to the fact that $a\in M_1^{\perp}\backslash M^\perp$. 
 So, {conditions} {(I)} and {(II)} of Theorem \ref{lm:gb 1 para} hold for $\tilde{a}\in {\mathbb Q}^s$. 

\begin{example}\label{ex:gb1 1}
We illustrate how Algorithm \ref{alg:gb1} works for $s=d=3$. 
\begin{description}
    \item[Step 1.] Enumerate all three-species zero-one reactions (there are $56$  reactions in total). 
    \item[Step 2.] Let $S:= \emptyset$. Note that $d=3$. For every $n\in \{1,2\}$, choose $n$ linearly independent reactions from the $56$ reactions enumerated in Step 1 and for every choice, do the following steps (below we only give details on the first round). 
    \begin{description}
        \item[Step 2.1] Suppose $(0,0,1)$ is the stoichiometric vector corresponding to the first reaction from the $56$ reactions. Let $\mathcal{M}:=(0,0,1)$. 
        \item[Step 2.2] The basis of the null space of $\mathcal{M}$ is $\{(1,0,0)^\top, (0,1,0)^\top\}$. Add the vectors in the basis into the set $S$. 
    \end{description} 
    Similarly, we  do Steps 2.1 and 2.2 for all the other choices of reactions.  
    Then, we get the set $S$: 
    \begin{align}
    \label{eq:3s3d S}
        &\{0,1\}^3 \cup (\{-1,1\}^ {{2}} \times\{2\}) \cup \nonumber\\
        &(\{-1,1\}\times\{-2,2\}\times\{1\}) \cup (\{-2,2\}\times\{-1,1\}\times\{1\}) \cup \nonumber\\
        &\{(-1,-1,1)^\top, (1,-1,1)^\top, (-1,1,1)^\top, 
        (-1,0,1)^\top,\nonumber\\
        &(-1,1,0)^\top,(0,-1,1)^\top\}\setminus (0,0,0)^\top. 
    \end{align} 
\end{description}
\end{example}

\section{Implementation}\label{sec:implement}
We have implemented the procedure outlined in Algorithm \ref{alg:main}, with supporting code publicly available at \href{https://github.com/YueJ13/equivalence}{https://github.com/YueJ13/equivalence}. We further apply this algorithm to all three-dimensional zero-one networks consisting of three species and five reactions.
\begin{description}
\item[Step 1.] Enumerate the stoichiometric coefficients $\alpha_{ij}$ and $\beta_{ij}$ for all $3$-species zero-one reactions and compute the corresponding stoichiometric vectors $\beta-\alpha$. Then enumerate all possible $3 \times 5$ stoichiometric matrices $\mathcal{N}$ of rank $3$ whose columns  are stoichiometric vectors. There are $43520$ matrices.  
        \item[Step 2.]  We call CoefficientsForInconsistent($3$, $3$, $5$), and by the computation carried out in Example \ref{ex:inconsistent}, we get a set of rational vectors $S_1$ in \eqref{eq:inconsistent}. For each stoichiometric matrix obtained in Step 1, we do the following steps. 
        \begin{description}
            \item[{Step 2.1.}]
        If there exists $a\in S_1$ such that the vector $a^\top \mathcal{N}$ is non-zero and does not change sign, then delete this matrix.
         \item[{Step 2.2.}]
        If the set $\mathcal{F}^+(\mathcal{N})$ is empty, then delete this matrix. 
        \end{description}
        Then, we get $14444$ matrices.
        \item[{Step 3.}] 
        For each remaining stoichiometric matrix,  delete the stoichiometric matrices that have the same form with it. 
        Then, we get $2548$ matrices.  
        \item[{Step 4.}] For every remaining stoichiometric matrix $\mathcal{N}$, generate all possible corresponding reactant matrices $\mathcal{Y}$. For every pair $\{\mathcal{N}, \mathcal{Y}\}$, we do the following steps.
                 \begin{description}
            \item[{Step 4.1.}] If $c=3$ or $4$, then delete the pair. ($\#$Note here, $c$ denotes the number of zero columns of the reactant matrix $\mathcal{Y}$.) 
        \item[{Step 4.2.}] We call CoefficientsForVacuous($3$, $3$), and by the compotation carried out in Example \ref{ex:gb1 1}, we get a set of rational vectors $S_2$ in \eqref{eq:3s3d S}.
        If there exists $a\in S_2$ such that
        \begin{enumerate}
            \item[{(I)}] $\mathcal{N}^\top a \neq \vec{0}$, and
            \item[{(II)}] for any $i \in \{1,2, 3\}$, the $i$-th coordinate of $\mathcal{N}^\top a$ is non-zero implies $col_i(\mathcal{Y})= \vec{0}$,
        \end{enumerate}
        then delete the pair. 
        \end{description}
       Then, we get $91802$ networks.
        \item[{Step 5.}] 
        Combine the remaining pairs $\{\mathcal{N}, \mathcal{Y}\}$ into the same group if they have  the same reduced row echelon form of $\mathcal{N}$. For each group, we do the following steps.    
        \begin{description}
            \item[{Step 5.1.}] For each pair $\{\mathcal{N}, \mathcal{Y}\}$ and $\{\hat{\mathcal{N}}, \hat{\mathcal{Y}}\}$, we do the following steps. 
            \begin{description}
                \item[{Step 5.1.1.}] Check if there exists a permutation on the rows of $\mathcal{Y}$ such that $\mathcal{Y}$ and $\hat{\mathcal{Y}}$ are equal after permuting the rows of $\mathcal{Y}$.   
          \begin{description}
            \item[{Step 5.1.1.1.}] If yes, then put the pair of matrices $\{\mathcal{N}, \mathcal{Y}\}$ and $\{\hat{\mathcal{N}}, \hat{\mathcal{Y}}\}$ into the same group. 
                \item[{Step 5.1.1.2.}] If not, then consider the next pair of matrices.
        \end{description}
         \end{description}
        \item[{Step 5.2.}] For every group, choose the first pair $\{\mathcal{N}, \mathcal{Y}\}$ as the representative element.  
        \end{description}
                 Then, we get $48703$ networks (pairs).
      \item[{Step 6.}]  For every remaining network, calculate the reduced Gr\"obner bases for steady-state ideal $\langle f \rangle$.  Delete the networks  whose corresponding reduced Gr\"obner bases are $\{1\}$ or have a monomial as an element. Then, we get $38249$ networks. 
      \item[{Step 7.}] Combine the networks with the same reduced Gr\"obner basis into the same class. Choose the first network as the representative element. Finally,  we get $38185$ networks.
    \end{description}

   We record the computational results for  the above  steps   in Table \ref{ta:main}. In the experiments, we  use  the command {\tt Basis} \cite{F1999} in {\tt Maple} \cite{maple2020} to calculate the reduced Gr\"obner bases. 
Note that  the number of all  three-dimensional  three-species 
five-reaction zero-one networks is  $3819816$ \cite{YX2025}.  If we do not apply  Algorithm \ref{alg:main},   then we need about $60$ hours to check the reduced Gr\"obner bases of these networks. Our main purpose is to demonstrate the improvements achieved by Algorithm \ref{alg:main} compared with the pure computation of Gr\"obner bases. The computational time of Step 7 presented in Table \ref{ta:main} can be improved by applying any more advanced tools to calculate the  Gr\"obner  bases rather than
the standard {\tt Maple} package. Recall that the original version of Algorithm \ref{alg:main}, as presented in \cite{YX2025}, requires approximately one hour to complete the entire process (see \cite[Table 1]{YX2025}). 
Here, the improved version  has saved about half of the total computational time by operating on stoichiometric matrices rather than networks in the first three steps.  Another observation is that most inconsistent networks are filtered out when  the vacuous ideals are verified at Step 4. 







\begin{table}[h]
\caption{Computational results for $(3,5,3)$-networks  (min: minutes; s: seconds)}
\label{ta:main}
\centering
\begin{tabular}{c|c|c}
\hline 
STEP&NUMBER OF Matrices&TIME  \\
\hline
Step 1 & 43520 & 1 min\\
Step 2 & 14444 & 1 min\\
Step 3 & 2548 & 1 min\\
\hline 
STEP&NUMBER OF NETWORKS&TIME  \\
\hline
Step 4 & 91802 & 1 min\\
Step 5 & 48703 & 30 s\\
Step 6 & 38249 & 2 min\\
Step 7 & 38185 & 22 min\\ 
\hline
\end{tabular}\\
{\bf Notes:} (i) The column ``STEP" lists the steps in the procedure shown in Algorithm \ref{alg:main}. (ii) The middle column records the number of networks we get after carrying out each step. (iii) The column ``TIME" records the computational time for  each step.\\
We run the procedure by a 3.60 GHz Inter Core i9 processor (64GB total memory) under Windows 10. 
\end{table}



In addition, we  apply Algorithm \ref{alg:main} to two key classes of 
 zero-one  networks:
 three-dimensional (with $3$ species and $6$ reactions) and four-dimensional  (with $4$ species and $5$ reactions). Given that both classes include a large number of networks (the numbers of matrices at Step 1 are $276382$ and $7818144$ respectively), we restrict our computations to quadratic networks only. Notably, Algorithm \ref{alg:main} can be easily adjusted to compute quadratic zero-one networks
 (in Step 1, we only enumerate the stoichiometric vectors for the quadratic reactions and in Step 4, we only consider the reactant matrices ${\mathcal Y}$ for quadratic networks).
 We record the timings for the completion of the steps of Algorithm \ref{alg:main} in Table \ref{tab:363} and Table \ref{tab:454}, respectively. 
 Another interesting finding is that
  there exists no networks after carrying out Step 6 for these two classes of quadratic networks.
 That means these networks have no positive steady states for a generic choice of rate constants. 
We show an example of three species six-reaction three-dimensional zero-one quadratic network remaining in Step 5, see Example \ref{ex:363}.
\begin{table}[h]
\caption{Computational results for $(3,6,3)$-quadratic networks  (min: minutes; s: seconds)}
\label{tab:363}
\centering
\begin{tabular}{c|c|c}
\hline 
STEP&NUMBER OF Matrices&TIME  \\
\hline
Step 1 & 10693 & 1 min\\
Step 2 & 5628 & 1 min\\
Step 3 & 960 & 1 min\\
\hline 
STEP&NUMBER OF NETWORKS&TIME  \\
\hline
Step 4 & 1897 & 1 min\\ 
Step 5 & 1565 & 0.2 s\\
Step 6 & \textcolor{red}{0} & 0.14 s\\
Step 7 & NULL & NULL\\ 
\hline
\end{tabular}\\
\end{table}
\begin{table}[h]
\caption{Computational results for $(4,5,4)$-quadratic networks  (min: minutes; s: seconds)}
\label{tab:454}
\centering
\begin{tabular}{c|c|c}
\hline 
STEP&NUMBER OF Matrices&TIME  \\
\hline
Step 1 & 1440168 & 3 min\\
Step 2 & 112008 & 1 min\\
Step 3 & 4822 & 5 min\\
\hline 
STEP&NUMBER OF NETWORKS&TIME  \\
\hline
Step 4 & 21440 & 1 min\\
Step 5 & 5032 & 1.9 s\\
Step 6 & \textcolor{red}{0} & 2.9 s\\
Step 7 & NULL & NULL\\ 
\hline
\end{tabular}\\
\end{table}

\begin{example}\label{ex:363}
  Consider the following $3$-dimensional quadratic zero-one network with $3$ species and $6$ reactions: 
\begin{align}
    &X_2+X_3 \xrightarrow{\kappa_1} X_1+X_3,&& X_1+X_3 \xrightarrow{\kappa_2}X_2+X_3, && X_1+X_3 \xrightarrow{\kappa_3}X_1+X_2, \nonumber\\ &X_1+X_2\xrightarrow{\kappa_4}X_1,&& X_1+X_2 \xrightarrow{\kappa_5}X_2+X_3, && X_1+X_2 \xrightarrow{\kappa_6}X_1+X_2+X_3. \nonumber
\end{align}
The steady-state system $f$ is 
\begin{align*}
    f_1 &= \kappa_1 x_2 x_3-\kappa_2 x_1 x_3-\kappa_5 x_1 x_2,\\
    f_2 &=-\kappa_1 x_2 x_3+\kappa_2 x_1 x_3+\kappa_3 x_1 x_3 -\kappa_4 x_1 x_2,\\
    f_3 &= -\kappa_3 x_1 x_3+\kappa_5 x_1 x_2+\kappa_6 x_1 x_2. 
\end{align*}
From $f_1 = f_2=f_3=0$, we can get 
\begin{align*}
   f_1+f_2+f_3=(-\kappa_4+\kappa_6)x_1x_2 = 0. 
\end{align*} 
Hence, if  $\kappa_4=\kappa_6$,  the network admits infinitely many positive steady states that satisfy
$$x_1=\frac{\kappa_1\kappa_3}{\kappa_3\kappa_5+\kappa_2\kappa_5+\kappa_2\kappa_6}x_3, \; x_2=\frac{\kappa_3}{\kappa_5+\kappa_6}x_3.$$
However, 
the reduced Gr\"obner basis of the steady-state ideal $\langle f \rangle$ in ${\mathbb Q}(\kappa)[x]$ is $$\{x_2x_3,\ x_1x_3,\ x_1x_2\}.$$ 
That means the network indeed has no positive steady states for any generic choice of $\kappa\in {\mathbb R}_{>0}^6$. 
So, this network is consistent and has a non-vacuous steady-state ideal—hence its passage through Steps 2 and 4 of Algorithm \ref{alg:main}—yet it is still dynamically trivial.
\end{example}

 {\bf Acknowledgements.
}We thank Professor Balázs Boros for his suggestion on how to efficiently generate small networks in the Joint Annual Meeting KSMB-SMB 2024. 
We thank Professor Murad Banaji  for generating and providing all three-dimensional three-species  five-reaction zero-one networks by his own efficient procedure as a reference for us. We would like to thank the anonymous reviewers of ISSAC 25 for their insightful and valuable comments, which have greatly contributed to the improvement of this work.



\end{document}